\documentclass[12pt]{article}

\usepackage{amsmath}
\usepackage{amssymb}
\usepackage{amsthm}

\usepackage{slashed}
\usepackage{color}

\usepackage{graphicx}
\usepackage{hyperref}

\newtheorem{thm}{Theorem}
\newtheorem{lemma}{Lemma}
\newtheorem{prop}{Proposition}
\newtheorem{cor}{Corollary}

\newtheorem{conj}{Conjecture}

\hypersetup{
    pdfborder = {0 0 0},%
    colorlinks,%
    citecolor=blue,%
    filecolor=black,%
    linkcolor=red,%
    urlcolor=green
}

\begin{document}

\title{  Renormalization group flow, entropy and eigenvalues }
\author{\textsc{Dan Li}
         \thanks{ E-mail: li1863@math.purdue.edu }
	 \\
           Department of Mathematics, Purdue University, \\
		150 N. University St, West Lafayette, IN 47907\\
	keywords: renormalization group flow, thermal entropy, \\ variation of eigenvalues, 
	          AdS/CFT correspondence        
	    }

\date{}
\maketitle

\begin{abstract}
The irreversibility of the renormalization group flow is conjectured to be closely related to the concept of entropy. 
In this paper, the variation of eigenvalues of the Laplacian in the Polyakov action under the renormalization group flow will be studied. 
Based on the one-loop approximation
to the effective field theory,
 we will use the heat kernel method and zeta function regularization.
In even dimensions, the variation of eigenvalues is given by the top heat kernel coefficient, and the conformal anomaly is relevant.
In odd dimensions, we will conjecture a formula for the variation of eigenvalues through the holographic renormalization in the setting of geometric AdS/CFT correspondence.

\end{abstract}

\section{Introduction}\label{Intro}

The idea of renormalization group (RG) emerges from the renormalization process of a quantum field theory (QFT) by which, if possible, all divergent infinities 
are taken care of and the QFT is made into a renormalizable theory. As a prototypical example, when the infinity in the momentum space is cut off by a
large energy scale $\Lambda$, the coupling constants are now depending on the introduced renormalization parameter $\Lambda$ and the original QFT becomes 
an effective field theory. In other words, one obtains different effective field theories parametrized by the cutoff scale in the space of quantum field theories (i.e., theory space).

As the energy scale varies, for example from the ultraviolet (UV) regime to the infrared (IR) regime, the renormalization group forms a renormalization group flow, which could have several fixed points.
One basic assumption of the RG flow is that if the change of scales is small from high energy to low energy, for example $\Lambda \rightarrow \Lambda' = e^{-t}\Lambda $, 
then the corresponding effective field theories are scale invariant. 
In most cases, scale invariance can be enhanced to conformal invariance \cite{N15}, which has important applications in geometry if the RG flow is approximated by a geometric flow such as the Ricci flow \cite{MT07}. 
Of course, the fixed points, i.e., some conformal field theories (CFTs) in the theory space, play important roles in the study of the RG flow. Another important 
observation of the RG flow is its irreversibility, which is conjectured to be intimately related to the concept of entropy. For example, 
for any renormalizable 2d field theory, the Zamolodchikov $c$-theorem \cite{Z86} tells us that the central charge at the fixed points of the RG flow 
decreases monotonically, i.e., $c_{UV} > c_{IR}$.

The beta function of a coupling parameter $g(\Lambda)$  depending on the energy scale $\Lambda$ is defined as 
\begin{equation}
\beta(g) :=  \Lambda \frac{\partial g}{\partial  \Lambda} = \frac{\partial g}{\partial \ln \Lambda}
\end{equation}
which describes the running of the coupling parameter $g$ along with the renormalization scale $\Lambda$.
At a fixed point of the RG flow, the conformal invariance implies the vanishing of the beta function, i.e., $\beta(g) = 0$.
In metric geometry, the coupling parameter $g$ is always taken as the metric $g_{ij}$ of a Riemannian manifold.  

Friedan first considered the perturbative RG flow of a two-dimensional nonlinear sigma model, 
and computed the beta function of the metric $g_{ij}$ of the target space based on Feynman diagram expansion \cite{F80, F85}.
The one-loop approximation of the RG flow, i.e., the first term in the beta function, gives rise to the Ricci flow,
\begin{equation}\label{RicFlow}
   \frac{\partial g_{ij}}{\partial t} = - 2R_{ij}
\end{equation}
where  $t = \ln (\Lambda^{-1})$ is the ``time'' parameter and $R_{ij}$ is the Ricci curvature.
As a remark, the parameter $t$ should not be confused with the physical time.
The Ricci flow equation \eqref{RicFlow},  which was first introduced into geometric analysis by Hamilton, is always viewed as the geometric analogy of the diffusive heat equation.
More recently, Perelman introduced the $W$-functional under the Ricci flow and proved the no local collapsing theorem
 in the resolution of the Poincar{\'e} conjecture and Thurston's geometrization conjecture \cite{MT07}. As pointed out by Perelman in \cite{P02},
 the definition of  $W$-functional was inspired by the concept of entropy in statistical mechanics, which is related to the 
 RG flow heuristically.
 
The variation of eigenvalues of geometric operators such as the Laplacian under the Ricci flow has been studied in the literature, 
 see \cite{ C08, D07, L07, M06}, 
and the monotonicity of eigenvalues  has been established. The idea behind such monotonicity is that
from the entropy functional such as the $F$-functional or $W$-functional by Perelman, eigenvalues of the corresponding geometric operators
inherit the monotonicity from the entropy functional. And such entropy functional originates from the thermal entropy in statistical mechanics, for example see \cite{L12}.

In this paper, we will first consider the variation of eigenvalues of the Laplacian in the Polyakov action under the RG flow
and express it in terms of spectral invariants, i.e., heat kernel coefficients and spectral zeta function.
In mathematical physics, the top heat kernel coefficient plays an important role, 
for example, in explanations of the conformal  anomaly \cite{V03}.
It is not a surprise that the variation of eigenvalues equals the top heat coefficient $a_{2k}$ in even dimensions, which is closely related to the conformal anomaly \cite{C88}. In general, there is no monotonicity
of eigenvalues or equivalently positivity of $a_{2k}$, since $a_{2k}$ consists of different parts expressed in scalar curvature, Ricci curvature and Riemann curvature tensor.
But sometimes it is possible to extract a positive or negative component such as the Euler density, the Weyl tensor or the Q-curvature so that the variation
of eigenvalues has a positive or negative sign under some mild conditions on the Riemannian manifold. For example, the 2d c-theorem \cite{Z86} and 4d a-theorem \cite{KS11} are important evidence of the existence of C-function, 
which is constructed to study the irreversibility of the RG flow. 


In odd dimensions, the situation is more subtle. We know that the odd heat kernel coefficients $a_{2i+1}$ vanish (for a closed manifold), and there is no conformal anomaly in odd-dimensional CFTs. 
From classical statistical mechanics, we learned that the partition function and the free energy can be used interchangeably to express the probability function.
If the $(2k+1)$-dimensional bulk theory is characterized by the free energy, then 
the $2k$-dimensional boundary theory should be described by the partition function.  For instance, in 3d there is an $F$-theorem \cite{CHM11}, where $F$ is the free energy,
and it is also viewed as an entanglement entropy. The AdS/CFT correspondence provides a geometric framework to study a duality equivalence between the bulk gravity theory and the boundary field theory \cite{GKP98, M98, W98}. 
The geometry of the AdS/CFT correspondence \cite{A05, PK05} allows us to translate  problems about the geometric operators from the bulk to some conformal invariants on the boundary. 
The idea of holographic renormalization \cite{DVV00, S02}, which links the radial coordinate to the energy scale, is developed to trade the IR divergence of the bulk gravity theory with the UV divergence of the boundary gauge theory.
So based on the holographic principle realized by the AdS/CFT correspondence, we conjecture a formula for the variation of eigenvalues of the bulk Laplacian in odd dimensions through the holographic RG flow.   


This paper is organized as follows. In Sect. \ref{Polact}, we will review some basic facts about the Polyakov action, perturbative RG flow and the heat kernel method.
The nonperturbative heat kernel approach will recover the first-order approximation to the RG flow, i.e., the Ricci flow.
In Sect. \ref{Varen}, we will derive the variation of eigenvalues of the Laplacian in the Polyakov action under the RG flow in even dimensions, which is given by the top heat kernel coefficient $a_{2k}$.
In Sect. \ref{Conjodd}, we will review the thermal entropy and the irreversibility of the RG flow, then a formula for the variation of eigenvalues of the bulk Laplacian in odd dimensions will be conjectured
based on the holographic principle. Finally, Sect. \ref{Disc} is some discussions about possible further investigations.

\section{Polyakov action}\label{Polact}

In this section, we first review some facts about the perturbative renormalization group (RG) flow and the heat kernel method. 
In the second part, we look into the Laplace operator derived from the Polyakov action and reproduce the Ricci flow from a heat kernel coefficient.  

\subsection{Perturbative RG flow}

In order to study the geometry of 
a string propagating in a curved background, one considers the nonlinear sigma model (or $\sigma$-model) defined by the Polyakov action, which is  
 the kinetic energy of a scalar field mapping a 2d surface $\Sigma$ 
into the target space $M$. More precisely, for a smooth embedding $\phi : (\Sigma, \gamma) \rightarrow (M, g)$,  
the Polyakov action is defined as  
\begin{equation}\label{SigmaMod}
       S( g; \phi) = \frac{T}{2}\int_\Sigma dv  \,\,\gamma^{\mu\nu}  \partial_\mu \phi^i (x) \partial_\nu \phi^j(x) g_{ij} (\phi(x))
\end{equation}
where $dv = \sqrt{|\gamma|} d\sigma$ is the standard volume measure on $\Sigma$. In string theory, the Polyakov action
describes a bosonic string model and 
the parameter $T$ is the string tension. Basically,
the Polyakov action computes the area of the string worldsheet in the curved target space. 
In the literature, the Regge slope $\alpha'$ is also widely used, and it is connected to $T$ 
by $ T = {1}/{2 \pi \alpha'} $. 
 A general nonlinear sigma model contains dilatons and B-fields etc.,
in this section we focus on the Polyakov action \eqref{SigmaMod}, i.e., the  free scalar field theory.

 For simplicity, we assume that the smooth Riemann surface $(\Sigma, \gamma)$ and the $n$-dimensional ($n \geq 2$) Riemannian manifold $(M,g)$ 
 are oriented and  compact without boundary. The space of smooth maps is denoted by $C^\infty(\Sigma, M)$,
 and the scalar fields $\phi : \Sigma \rightarrow M $ are assumed to be smooth. The Polyakov action 
  $S(g; \phi)$ is also assumed to be invariant under the diffeomorphism group of $M$, i.e., $Diff(M)$. 
In addition, we only consider Euclidean quantum field theory (QFT), and the metric $g$ has Euclidean signature.

The RG flow of the nonlinear  sigma model \eqref{SigmaMod} was first investigated by Friedan \cite{F80, F85},
an approach based on dimensional regularization of Feynman diagram evaluation can be found in  \cite{G00}, 
and the geometric picture of the RG flow is explained  in \cite{C10}.

We briefly recall the idea of perturbative RG flow based on regularization and renormalization following \cite{C10}.
First one considers the effective action $W$ defined by 
\begin{equation}
  e^{-  W } = \int [D\phi] e^{- S( g; \phi)} 
\end{equation} 
where the path integral measure is formally defined over scalar fields modulo diffeomorphic equivalence.
When the parameter $T \rightarrow \infty$ or $\alpha' \rightarrow 0$, the effective action is roughly the classical action plus quantum 
fluctuations around the stationary points of the action $S( g; \phi)$.
If there exists a decomposition  $\phi= \phi_0 + \pi$ around a given extremizing field $\phi_0$ such that 
$\frac{\delta S}{\delta \phi}(\phi_0) = 0$ and $\pi \in C^\infty(\Sigma, \phi^*TM)$ is a small perturbation (or quantum fluctuation), 
then under the saddle point approximation
the local contribution to $W$ around  $\phi_0$ is approximated by
\begin{equation}
   W(\phi_0) \sim S(\phi_0) - ln \int [D\pi] exp\left(-\frac{1}{2} S''_{\mu\nu}(\phi_0) \pi^\mu \pi^\nu  \right)
\end{equation}
Using the normal geodesic coordinates around $\phi_0$, $S''_{\mu\nu}(\phi_0)$ is expanded out and the Riemann curvature 
$R_{ijkl}$ appears naturally in the covariant derivatives. Applying the Feynman rules to one-loop Feynman diagrams,  a momentum cutoff $\Lambda$ is introduced into 
the Feynman propagator. After this regularization, the divergent part of the field theory   
only lies in a term with the  logarithmic energy $\ln \Lambda$ multiplied by  the Ricci curvature $R_{ij}$. 
Finally, the renormalization procedure identifies
the Ricci curvature as  the lowest order term in the beta function.
More precisely, the first-order term in $\alpha'$ in the RG flow equation gives rise to the Ricci flow,
\begin{equation}
   \frac{d g_{ij}}{d t} =-2\alpha' R_{ij} + O(\alpha'^2)
\end{equation}
where  $t =\ln \Lambda^{-1}$ or  $\Lambda = e^{-t}$. 
 
The higher-order terms in the RG flow can be obtained similarly by evaluating Feynman diagrams with more loops, 
and the first two terms were computed in Friedan's thesis \cite{F85}, 
\begin{equation}
  \frac{d g_{ij}}{d t} = -\beta(g_{ij}) =-2\alpha' R_{ij}- {\alpha'^2}R_i^{\, klm}R_{jklm}+ O(\alpha'^3)
\end{equation}

The perturbative RG flow is conjectured to be a good approximation to the nonperturbative RG flow. 
In the next section, we will look at a nonperturbative approach based on the heat kernel method.
Because of the feature of heat kernel expansion, this approach only determines the
one-loop approximation to the effective field theory.

\subsection{Heat kernel expansion}

Suppose $(M, g)$ is a compact smooth Riemannian manifold without boundary, let $\Delta$ be the Laplace--Beltrami operator with positive eigenvalues
\begin{equation}
   0 <  \lambda_1 \leq \lambda_2 \leq \cdots \lambda_k \leq \cdots
\end{equation}
for simplicity, here the eigenvalue $\lambda_0 = 0$ is ignored. 
If the corresponding eigenfunctions are given by $\{ \psi_k\}$, which form an orthonormal basis for $L^2(M)$, 
then the heat kernel $K(t; x, y; \Delta)$ is defined as
\begin{equation}
   K(t; x, y; \Delta) = \sum_k e^{-\lambda_k t} \psi_k(x)\psi_k(y)
\end{equation}
The trace of the heat kernel, sometimes called the partition function $Z(t)$, is defined by 
\begin{equation}\label{Trheat}
  Z(t)= Tr(e^{-t\Delta}) = \sum_k e^{-t\lambda_k} = \int_M dV\, K(t; x, x; \Delta) = K(t, \Delta)
\end{equation}
which has an asymptotic expansion when $t$ tends to zero,
\begin{equation}
  Tr(e^{-t\Delta}) \sim  (4\pi t)^{-\frac{n}{2}}(a_0 + a_2t + a_4t^2 + \cdots )  \quad \text{as} \quad  t \rightarrow 0^+
\end{equation}
where $\dim M = n$ and the heat kernel coefficients $a_{2i}$ are integrals of local spectral invariants,
for example, 
\begin{equation}
  \begin{array}{l}
   a_0 = Vol(M, g)\\
   a_2 = \frac{1}{6 }\int_M dV  \,R   \\
   a_4 = \frac{1}{360}\int_M dV \,(-12\Delta R+5R^2 -2 R^{ij}R_{ij} +2 R^{ijkl}R_{ijkl})   \\
  \end{array}
\end{equation}

One defines the spectral zeta function of the Laplace--Beltrami operator $\Delta$ as usual,
\begin{equation} \label{zetafun}
   \zeta_{\Delta}(s) := \sum_{k = 1}^\infty  \lambda_k^{-s} = Tr(\Delta^{-s})
\end{equation}
which is well-defined for $Re(s) > \frac{n}{2}$.
The zeta function is related to the heat kernel by the Mellin transform,
\begin{equation}
   \zeta_{\Delta}(s) = \frac{1}{\Gamma(s)} \int_0^\infty t^{s-1} K(t, \Delta) dt = \frac{1}{\Gamma(s)} \int_0^\infty t^{s-1} Tr(e^{-t\Delta}) dt 
\end{equation}
where $\Gamma(s)$ is the gamma function.

From the basic formula for a positive eigenvalue $\lambda > 0$ of $\Delta$, 
\begin{equation}
   \ln \lambda = - \int_0^\infty \frac{dt}{t} e^{-t\lambda}
\end{equation}
the  logarithmic determinant (or log-determinant) of $\Delta$ can be expressed as
\begin{equation}
   \ln \det (\Delta) = Tr \ln (\Delta) =  - \int_0^\infty \frac{dt}{t} Tr(e^{-t\Delta} )
\end{equation}
The one-loop effective action is defined by 
\begin{equation}
   W := \frac{1}{2}  \ln \det (\Delta) = -  \frac{1}{2}  \int_0^\infty \frac{dt}{t} Tr(e^{-t\Delta} )
\end{equation}

The above effective action $W$ also appears in the zeta function regularization scheme, more details can be found in \cite{V03}, if the regularized effective action is defined as
\begin{equation}
   W_s := -\frac{1}{2} \tilde{\mu}^{2s} \Gamma(s)  \zeta_{\Delta}(s) = -\frac{1}{2} \tilde{\mu}^{2s} \int_0^\infty t^{s-1} Tr(e^{-t\Delta}) dt 
\end{equation}
where $\tilde{\mu}$ is a constant parameter,
then $W$ is the limit $W = \lim_{s \rightarrow 0} W_s$, and the regularization is removed in the limit  $s \rightarrow 0$. 
In this context,  $ \zeta_{\Delta}(s)$ is also called the regularized zeta function. 
Since the gamma function $\Gamma(s)$ has a simple pole at $s= 0$, the regularized effective action $W_s$ also has a pole at $s= 0$.
By renormalization, the pole at $s= 0$ can be removed and  the renormalized effective action is obtained as
\begin{equation} \label{RenEffAct}
   W_{ren}(\mu) = -\frac{1}{2}\zeta_\Delta'(0) -\frac{1}{2}\ln (\mu^2) \zeta_\Delta(0)
\end{equation}
where the rescaled parameter $\mu$ is related to $\tilde{\mu}$ by $\mu^2 = e^{-\gamma_E} \tilde{\mu}^2$ ($\gamma_E$ is the Euler constant).

Without the regularization and renormalization, the effective action $W$ is expressed in terms of the  zeta function as
\begin{equation}
  W =    - \frac{1}{2} \frac{d}{ds} \zeta_\Delta(s)|_{s= 0} = -  \frac{1}{2} \zeta_\Delta'(0)
\end{equation}
from the observation (differentiating the series form of $\zeta_\Delta(s)$  \eqref{zetafun} term by term)
\begin{equation}
   \frac{d}{ds} \zeta_\Delta(s) =  \sum_{k = 1}^\infty  \frac{- \ln \lambda_k}{\lambda_k^{s}} 
\end{equation}

For a general Laplace-type operator $D = - \Delta + V$, where $V$ is a matrix-valued potential, 
and an auxiliary smooth matrix-valued function $f$, the trace of the heat kernel is well defined,
\begin{equation}
    Tr(f e^{-tD}) = \int_M dV\, K(t; x, x; D) f(x)  = K(t, f, D)
\end{equation}
There is an asymptotic expansion when $t$ tends to $0$,
\begin{equation}
   Tr(f e^{-tD}) \sim (4\pi t)^{-\frac{n}{2}}[a_0(f, D) + a_2(f, D)t + a_4(f, D)t^2 + \cdots ]  \quad  t \rightarrow 0^+
\end{equation}
For example, the first three nonzero terms are given by 
\begin{equation}
  \begin{array}{l}
   a_0(f, D) = \int_M dV ~Tr\{f\} \\
   a_2(f, D) = \frac{1}{6 }\int_M dV  \,Tr\{f(6V + R)\}   \\
   a_4(f, D) =  \frac{1}{360}\int_M dV \, Tr\{ f(60\Delta V + 60RV + 180V^2  \\
  \qquad \qquad \qquad + 12\Delta R+ 5R^2 -2 R^{ij}R_{ij} +2 R^{ijkl}R_{ijkl} ) \}  
  \end{array}
\end{equation}

Similarly, the spectral zeta function can be generalized as
\begin{equation}
   \zeta(s, f, D)= Tr(fD^{-s})
\end{equation}
and by the Mellin transform,
\begin{equation}
   \zeta(s, f, D) = \frac{1}{\Gamma(s)} \int_0^\infty dt \, t^{s-1} Tr(f e^{-tD})
\end{equation}
The heat kernel coefficients can be expressed as the residues at the poles,
\begin{equation}
   a_k(f,D) = Res_{s = (n-k)/2} (\Gamma(s) \zeta(s, f, D))
\end{equation}
in particular, 
\begin{equation}
   a_n(f, D) = \zeta(0, f, D)
\end{equation}

If $N$ is a smooth compact Riemannian manifold with smooth boundary $\partial N$, then one has to define boundary conditions $\mathcal{B}$, for example, Dirichlet or Neumann boundary conditions.
There also exists an asymptotic expansion,
\begin{equation}
    Tr(f e^{-tD}) \sim \sum_{k \geq 0}  t^{\frac{(k-n)}{2}}a_k(f, D, \mathcal{B})    \quad  t \rightarrow 0^+
\end{equation}
where $a_k(f, D, \mathcal{B})$ are local invariants depending on the boundary data  $\mathcal{B}$, and the odd terms $a_{2i+1}(f, D, \mathcal{B})$ are nonvanishing.

\subsection{Laplace operator}

The Polyakov action \eqref{SigmaMod} is also written as 
\begin{equation}
     S( g; \phi)  =  \int_\Sigma dv \,tr_{\gamma} (\phi^*g) 
\end{equation}
where  the string tension $T$ has been absorbed into the action and the Lagrangian is the trace of the pullback or induced metric,  
\begin{equation}
   \phi^*g_{\mu \nu}(x) = g_{ij}(\phi(x))  \partial_\mu \phi^i(x) \partial_\nu \phi^j(x) 
\end{equation}

In particular, if the background metric $g$ is  flat, i.e., $g_{ij} = \delta_{ij}$, then one obtains the simplified action, 
\begin{equation}
     S( \delta; \phi)  =   \int_\Sigma d\sigma \,\sqrt{|\gamma|} \gamma^{\mu\nu}  \partial_\mu \phi^i (x) \partial_\nu \phi_i(x) 
\end{equation}
If one further applies integration by parts, then the action becomes
\begin{equation}
     S( \delta; \phi)  =  - \int_\Sigma dv \,  \phi^i (x)  \Delta_\gamma \phi_i(x) 
\end{equation}
where $\Delta_\gamma$ is the Laplace--Beltrami operator with respect to $\gamma_{\mu\nu}$,
\begin{equation}
  \Delta_\gamma u = \frac{1}{\sqrt{|\gamma|}} \partial_\mu ( \sqrt{|\gamma|} \gamma^{\mu\nu}    \partial_\nu u )
\end{equation}
In practice, the  metric $g$ is sometimes decomposed as $g_{ij} = \delta_{ij} + h_{ij}$; in this case, the Polyakov action is written as
\begin{equation}
   S(g; \phi) = S(\delta; \phi) + S(h; \phi)
\end{equation}
where $S(h; \phi)$ is viewed as a fluctuation.

\begin{lemma}
   The Polyakov action \eqref{SigmaMod} can be expressed as 
   \begin{equation}\label{Pact}
            S( g; \phi)  =     -  \int_\Sigma dv \,  g_{ij}\phi^i \partial^\mu  \partial_\mu \phi^j + \frac{1}{2} \int_\Sigma dv \,  ( \partial^\mu  \partial_\mu g_{ij} ) \phi^i \phi^j
   \end{equation}
\end{lemma}

\begin{proof}
The main tool here is  integration by parts. 
\begin{equation}
   \begin{array}{ll}
     S( g; \phi) & = \int_\Sigma dv \,   g_{ij}(\phi) \partial^\mu\phi^i  \partial_\mu \phi^j  \\
                 & = -  \int_\Sigma dv \,  \phi^i \, \partial^\mu [g_{ij}(\phi)   \partial_\mu \phi^j] \\
                 & = -  \int_\Sigma dv \, \phi^i g_{ij}(\phi)  \partial^\mu  \partial_\mu \phi^j -  \int_\Sigma dv \, \phi^i \partial^\mu g_{ij}(\phi)     \partial_\mu \phi^j   \\
                 
   \end{array}
\end{equation}
where the modified Laplace operator will be denoted by $\Delta_g  = g_{ij}(\phi)  \partial^\mu \partial_\mu$. 
Let us continue to work on the second integral,
\begin{equation}
   \begin{array}{ll}
     - \int_\Sigma dv \, \phi^i  \partial^\mu g_{ij}(\phi)  \partial_\mu \phi^j 
     & =   \int_\Sigma dv \,  \phi^j  \partial_\mu [  \phi^i \partial^\mu g_{ij}(\phi)]   \\
     & =    \int_\Sigma dv \,  \phi^j  \partial_\mu   \phi^i \partial^\mu g_{ij}(\phi)  + \int_\Sigma dv \,  \phi^j    \phi^i  \partial_\mu \partial^\mu g_{ij}(\phi)   \\
       \end{array}
\end{equation} 
The above first integral is 
\begin{equation}
   \begin{array}{ll}
      \int_\Sigma dv \,  \phi^j  \partial_\mu   \phi^i \partial^\mu g_{ij}(\phi)
     & =  - \int_\Sigma dv \, g_{ij}(\phi) \partial^\mu [\phi^j  \partial_\mu   \phi^i]      \\
     & =   -\int_\Sigma dv \, g_{ij}(\phi)\partial_\mu   \phi^i  \partial^\mu \phi^j     -\int_\Sigma dv \, g_{ij}(\phi) \phi^j \partial^\mu  \partial_\mu \phi^i      
               
   \end{array}
\end{equation} 

Putting it together, we obtain
\begin{equation}
  \begin{array}{ll}
   S( g; \phi)  = &  - 2 \int_\Sigma dv \,  \phi^i \Delta_g \phi^j -  S( g; \phi) 
           + \int_\Sigma dv \,   \phi^i  \phi^j \partial_\mu \partial^\mu g_{ij}(\phi)

  \end{array}
\end{equation}
Or equivalently, 
\begin{equation}
   \begin{array}{ll}
      S( g; \phi) & =     -  \int_\Sigma dv \,  \phi^i \Delta_g \phi^j + \frac{1}{2} \int_\Sigma dv \,   \phi^i  \Delta g_{ij}(\phi)  \phi^j
   \end{array}
\end{equation}
where $\Delta = \gamma^{\mu\nu}\partial_\nu \partial_\mu = \partial^\mu \partial_\mu  $. If the background metric $g$ is flat, 
i.e., $g_{ij} = \delta_{ij}$, then the second term is zero and the first term recovers the Laplace--Beltrami operator.

\end{proof}

If the real scalar field $\phi: \Sigma \rightarrow M$ has the local coordinates $\phi(x) = (\phi^1(x), \cdots, \phi^n(x) )$ for $x \in \Sigma$,
then we can denote the above integrals by  inner products, 
\begin{equation}
   \langle \phi , \Delta \phi \rangle_g := \int_\Sigma dv \,  g_{ij}\phi^i \partial^\mu  \partial_\mu \phi^j 
\end{equation}
\begin{equation}
   \langle \phi , \phi \rangle_{\Delta g} := \int_\Sigma dv \,  ( \partial^\mu  \partial_\mu g_{ij} ) \phi^i \phi^j
\end{equation}
Similarly, the Polyakov action can be expressed as 
\begin{equation}
  S(g; \phi) =   \int_\Sigma dv \,g_{ij}  \partial^\mu \phi^i \partial_\mu \phi^j =   \langle d\phi , d \phi \rangle_g 
\end{equation}

\begin{lemma}
   The relation \eqref{Pact} in the above lemma is equivalent to
   \begin{equation}
      \langle d\phi , d \phi \rangle_g  = - \langle \phi , \Delta \phi \rangle_g + \frac{1}{2} \langle \phi , \phi \rangle_{\Delta g}
   \end{equation}
or under the standard inner product  $ \langle \phi ,  \phi \rangle := \sum_{i,j=1}^n  \int_\Sigma dv \,    \phi^i \phi^j$,  
    \begin{equation}
      \langle d\phi ,g d \phi \rangle  = - \langle \phi , g\Delta \phi \rangle + \frac{1}{2} \langle \phi , {\Delta g}\phi \rangle
   \end{equation}
\end{lemma}

Therefore, the Laplace operator in the Polyakov action is given by 
\begin{equation} \label{LapOp}
   L(\phi, g, \gamma ) = -g\Delta + \frac{1}{2}{\Delta g} = - g_{ij}(\phi)\gamma^{\mu\nu}\partial_\mu\partial_\nu + \frac{1}{2} \gamma^{\mu\nu}\partial_\mu\partial_\nu  g_{ij}(\phi)
\end{equation}
If we view $g_{ij}$ as an auxiliary function, then in the standard form
\begin{equation}
    L(\phi, g, \gamma )  = g(-\Delta + V) = gD(\phi, g, \gamma)
\end{equation}
where the potential $V = \frac{1}{2}g^{-1}\Delta g$ and $D = -\Delta +V$ is a Laplace-type operator. 
For example, if $g_{ij}(\phi) = e^{\phi(x)}g^0_{ij} $ 
for some constant metric $g^0$ in the background, then 
\begin{equation}
 D = -\Delta + \frac{1}{2}(\partial^\mu\phi \partial_\mu\phi + \Delta \phi) 
\end{equation}

\subsection{Ricci flow}

Now we apply the heat kernel method to the operator $D$ and the auxiliary function $g$,  we have the partition function depending on $g$,
\begin{equation}
 Z(t, g) = Tr(ge^{-tD}) 
\end{equation}
and the one-loop effective action of the Polyakov action,
\begin{equation}
   W(g) = \frac{1}{2} \ln \det  (gD) = - \frac{1}{2}\int_0^\infty \frac{dt}{t} Tr(ge^{-tD})
\end{equation}

Suppose  the coupling parameter $g$ is a function of the energy scale $\Lambda$, then the Callan--Symanzik equation reads
\begin{equation}
  \frac{\partial W(g(\Lambda)) }{\partial \ln \Lambda} + \beta(g) \frac{\partial W(g)}{\partial g} = 0
\end{equation}
So the beta function can be computed by
\begin{equation}
 \beta(g) = - \frac{\partial W(g(\Lambda)) }{\partial \ln \Lambda} / \frac{\partial W(g)}{\partial g}
\end{equation}
If we differentiate the effective action $W(g)$, then  the denominator is 
\begin{equation}
   \frac{\partial W(g)}{\partial g}  = - \frac{1}{2}\int_0^\infty \frac{dt}{t} Tr(e^{-tD})
\end{equation} 
which can be viewed as a normalization constant $W = \frac{1}{2} \ln \det  (D)$ independent of $g$.

From \cite{H77} (or Eq. \eqref{RenEffAct}), the renormalized effective action depending on the renormalization parameter $\Lambda$ (here we use energy scale $\Lambda$ instead of $\mu$) is given by
\begin{equation}
  W(g (\Lambda)) = - \frac{1}{2} [\zeta'(0, gD) + \ln \Lambda^2 \, a_n(g, D)] 
\end{equation}
In order to approximate the first-order term in the beta function, it is enough to take $a_2(g, D)$,
so we have
\begin{equation}
   -\beta(g) =  \frac{1}{W} \frac{\partial W(g(\Lambda)) }{\partial \ln \Lambda} \sim - a_2(g, D) = -\frac{1}{6 }\int_M dV  \,Tr(g R) -\frac{1}{2} \int_M dV  \,Tr ( \Delta g) 
\end{equation}
If we view the first integral in $a_2(g, D)$ as a finite number, and the local form of the second integral recovers the Ricci flow,
\begin{equation}
  \frac{\partial g_{ij}}{\partial t} = -R_{ij} 
\end{equation}
since one has the approximation $\Delta g_{ij} \sim -2R_{ij}$ in  harmonic  coordinates.

\section{Variation of eigenvalues in even dimensions}\label{Varen}

In this section, we will continue to consider the Polyakov action \eqref{SigmaMod}, and derive a formula for
the variation of eigenvalues of the Laplace operator \eqref{LapOp} under the renormalization group (RG) flow.

Let us first briefly review the scale invariance, which is the key in understanding  the RG flow, see \cite{C10} for more details. 
The RG flow  can be represented by the RG transformation, 
\begin{equation}
   \mathcal{RG}_t: C^\infty(\Sigma, M) \times \mathcal{M}  \rightarrow C^\infty(\Sigma, M) \times \mathcal{M} ; \quad (\phi, g) \mapsto (\phi_t, g_t)
\end{equation}
where $t$ is a length scale and $\mathcal{M}$ is the coupling space of metrics. Moreover, the RG transformation  on the space of effective actions is
\begin{equation}
   \begin{array}{rll}
         \widetilde{\mathcal{RG}}_{\Lambda \Lambda'}: \mathcal{ACT}(\Sigma, M; \mathcal{M}) & \rightarrow  & \mathcal{ACT}(\Sigma, M; \mathcal{M}) \\
                                              S(\phi_\Lambda; g(\Lambda)) & \mapsto &  S(\phi_{\Lambda'}; g(\Lambda'))
   \end{array}
\end{equation}
where $\Lambda > \Lambda'$ are energy scales  and the length scale $t = \ln(\Lambda/\Lambda')$ (or $\Lambda' = e^{-t}\Lambda$), note that $\mathcal{RG}_{\Lambda \Lambda'} =  \mathcal{RG}_{-t}$.
The following commutation relation is a natural requirement on the RG transformation,
\begin{equation}
    \widetilde{\mathcal{RG}}_{\Lambda \Lambda'}   S(\phi_\Lambda; g(\Lambda)) =   S(  \mathcal{RG}_{\Lambda \Lambda'}(\phi_\Lambda; g(\Lambda)))
\end{equation}
The measure in the path integral is defined by
\begin{equation}
   d\rho (\phi; g) := D\phi \, e^{-S(\phi; g)}
 \end{equation} 
and we have the measure space $\{C^\infty(\Sigma, M) \times \mathcal{M},  d\rho (\phi; g) \}$. In addition, the RG transformation acts on the measure by pullback
\begin{equation}
    \mathcal{RG}^*_t ( d\rho (\phi; g) )   =  \mathcal{RG}_{-t} (D\phi) e^{- \widetilde{ \mathcal{RG}_{-t}} S(\phi; g) } 
\end{equation}
The scale invariance requires that the QFTs are invariant under the change of scale by $t$ if $t$ is infinitesimal ($0 < t < \varepsilon$), that is, 
\begin{equation}
   \int_{C^\infty(\Sigma, M) \times \mathcal{M} }   d\rho (\phi; g) = \int_{  \mathcal{RG}_t  (C^\infty(\Sigma, M) \times \mathcal{M}) }  \mathcal{RG}^*_t ( d\rho (\phi; g) ) 
\end{equation}
In other words, along the RG flow, the distance scale becomes larger ($\Lambda^{-1} \rightarrow \Lambda'^{-1} $) and the measure becomes smaller, so the effective QFTs look the same.

\begin{thm} By the scale invariance of the RG flow, the nonperturbative RG flow equation of the eigenvalues of the Laplace operator in the Polyakov action  is  
   \begin{equation}
      \frac{\partial \lambda  }{\partial t}  =  a_n({g}, D) = \zeta(0, g, D)
   \end{equation}
   Here the minimal subtraction renormalization scheme is used; of course, the result as  a one-loop approximation is independent of the choice of renormalization schemes. 
\end{thm}

\begin{proof}
 
For a scalar field with upper  energy bound $\Lambda$, i.e., $\phi \in C^{\infty}(\Sigma, M) [0, \Lambda]$, 
the effective action $S_{eff}[\Lambda, g]$ is defined by
 \begin{equation}
   e^{- S_{eff}[\Lambda, g]  } = \int_{C^{\infty}(\Sigma, M) [0, \Lambda]} D\phi \,\,exp\left\{ - \langle {\phi}, L(\Lambda, \phi, g ) {\phi} \rangle \right\} 
\end{equation}
where $L(\Lambda, \phi, g) = L(\phi, g, \gamma)$ from \eqref{LapOp} is the  Laplace operator in the Polyakov action, 
and now it is assumed to be depending on $\Lambda$. 
Decompose the scalar field into $\phi = \psi + \chi$ such that  $\psi \in C^{\infty}(\Sigma, M) [0, \Lambda']$ 
and $\chi \in C^{\infty}(\Sigma, M) ( \Lambda', \Lambda]$, the path integral measure factorizes accordingly, i.e., $D\phi = D\psi\, D\chi$, so  
\begin{equation}
         e^{- S_{eff}[\Lambda, g]  } = \int_{[0, \Lambda']} D\psi \int_{ (\Lambda', \Lambda]} D\chi \,\, exp\left\{ - 
         \langle {\psi+ \chi}, L(\Lambda,  {\psi+ \chi}, g) ({\psi + \chi }) \rangle \right\}
\end{equation}
The action functional in the right hand side can be decomposed as
\begin{equation}
   \begin{array}{ll}
    &   \langle {\psi+ \chi}, L(\Lambda, \psi +\chi, g) ({\psi + \chi }) \rangle \\
    & =   \langle {\psi}, L(\Lambda, \psi, g) {\psi } \rangle +    \langle {\chi}, L(\Lambda, \chi, g) { \chi } \rangle + S^{int}_\Lambda[\psi, \chi]
   \end{array}   
\end{equation}
where $S^{int}_\Lambda[\psi, \chi]$ is the effective interaction term.  $S^{int}_\Lambda[\psi, \chi]$ can be ignored since $\psi$ and $\chi$ fall into distinct energy range, 
so we rewrite the above path integral as
\begin{equation}
         \int_{[0,{\Lambda'}]} D\psi exp\left\{ -  \langle {\psi}, L(\Lambda, \psi, g) {\psi } \rangle \right\} \int_{[{\Lambda'}, \Lambda]} D\chi  
            exp\left\{ - \langle {\chi}, L(\Lambda,\chi,g) { \chi } \rangle \right\}
\end{equation}
Denote the integrand using a  new effective action  $S(\psi, g)$,
\begin{equation}\label{RGlogic}
   e^{- S(\psi,g)} = exp\left\{ - \langle {\psi}, L(\Lambda,\psi, g) {\psi } \rangle \right\} \int_{[\Lambda', \Lambda]} D\chi \, 
              exp\left\{ - \langle {\chi}, L(\Lambda, \chi, g) { \chi } \rangle \right\}
\end{equation}
According to the scale invariance of the RG flow, $S(\psi, g)$ 
must take the form $   \mathcal{RG}_{\Lambda \Lambda'}S(\phi_{\Lambda}, g_{\Lambda}) =  S(\phi_{\Lambda'}, g_{\Lambda'})$, namely
\begin{equation} \label{RGAsp}
  S(\psi,g) =  \langle {\psi}, L(\Lambda',\psi, g) {\psi } \rangle 
\end{equation}
 with the energy scale $\Lambda'$ in the Laplacian $L$.
This assumption guarantees the  effective QFTs are invariant under the change of scale from $\Lambda$ to $\Lambda'$,
\begin{equation}
   \begin{array}{ll}
        \int_{[0, \Lambda]} D\phi \,\,exp\left\{ - \langle {\phi}, L(\Lambda, \phi, g) {\phi} \rangle \right\} 
                                     = \int_{[0, {\Lambda'}]} D\psi \,\,exp\left\{ - \langle {\psi}, L(\Lambda', \psi, g) {\psi} \rangle \right\} 
                                      
   \end{array}
\end{equation}

Now under the scale invariance \eqref{RGAsp},  we rewrite   \eqref{RGlogic}  as
\begin{equation}
    \langle {\psi}, L(\Lambda', \psi ,g) {\psi }\rangle =   \langle {\psi}, L(\Lambda, \psi, g) {\psi } \rangle - 
    \ln \int_{\Lambda'}^\Lambda D\chi \, 
              exp\left\{ - \langle {\chi}, L(\Lambda, \chi, g) { \chi } \rangle \right\}
\end{equation}
Suppose that $ L(\Lambda, \psi, g) \psi = \lambda(\Lambda) \psi$ at the energy scale $\Lambda$ for an eigenfunction $\psi$,
and $ L(\Lambda', \psi, g) \psi = \lambda(\Lambda') \psi$ is assumed to be the effective
eigenvalue equation at $\Lambda'$ for the same eigenfunction $\psi$.
Note that the eigenfunction $\psi$ can be chosen such that $ \langle {\psi}, {\psi }\rangle = 1$, and the definite integral  is  represented by
the truncated Laplacian determinant, so we obtain
\begin{equation}\label{Lapdet}
   \lambda(\Lambda') = \lambda(\Lambda) + \ln \det  (gD)|_{\Lambda'}^\Lambda
\end{equation}

Using the heat kernel, it can be expressed as
\begin{equation}\label{Heateq}
    \lambda(\Lambda') = \lambda(\Lambda) - \int_{\Lambda^{-2}}^{\Lambda'^{-2}} \frac{dt}{t} K(t,g, D)
\end{equation}
where the integral variable is identified to be $t \sim \Lambda^{-2}$. 
Suppose that the energy level $\Lambda$ is large and $\Lambda'$ is very close to $\Lambda$ 
so that $0 < \Lambda^{-2} < \Lambda'^{-2} < 1$; in this case, the heat kernel is asymptotically expanded as usual
\begin{equation}
   K(t, g, D) \sim \sum_{k \geq 0} t^{(k-n)/2} a_k(g, D) 
\end{equation}
where the multiple constant $(4\pi)^{-n/2}$ is absorbed in $a_k$. 

Plug the asymptotic expansion into  \eqref{Heateq}, we get
\begin{equation}
    \lambda(\Lambda') = \lambda(\Lambda) - \sum_{k \geq 0} a_k(g, D) \int_{\Lambda^{-2}}^{\Lambda'^{-2}}  t^{(k-n-2)/2}  {dt}
\end{equation}
There are three different parts in the above summation after we integrate it out. 
Here we introduce a new variable 
\begin{equation}
   \tau = \ln (\Lambda'^{-2}/\Lambda^{-2}), \quad i.e. \quad \Lambda'^{-2} = e^\tau \Lambda^{-2} 
\end{equation}

$\bullet$  $ 0 \leq k < n$, the first part is divergent if $\Lambda \rightarrow \infty $,
\begin{equation}
   \sum_{0 \leq k < n} \frac{2a_k}{(k-n)}t^{(k-n)/2}|_{\Lambda^{-2}}^{\Lambda'^{-2}} 
   = \sum_{0 \leq k < n} \frac{2a_k}{(k-n)}{\Lambda}^{n-k}( e^{-\tau(n-k)/2}- 1)  
\end{equation}

$\bullet$  $k> n$, this  is the convergent part in the limit $\Lambda \rightarrow \infty $,
\begin{equation}
   \sum_{k > n} \frac{2a_k}{(k-n)}t^{(k-n)/2}|_{\Lambda^{-2}}^{\Lambda'^{-2}}
   = \sum_{ k > n} \frac{2a_k}{(k-n)}\frac{e^{\tau(k-n)/2}- 1}{\Lambda^{k-n}}  
\end{equation}

$\bullet$  $k = n$, this part has the expected logarithmic growth,
\begin{equation}
  2\ln (\Lambda/\Lambda') \,a_n(g, D)
\end{equation}

We pick out this logarithmic growth term as the critical part; 
meanwhile, we have to subtract the divergent part and ignore the convergent part.
After such  regularization and renormalization,  we get the expression of the effective eigenvalue at the energy scale $\Lambda'$,
\begin{equation}
   \lambda(\Lambda') \sim \lambda(\Lambda) - 2\ln (\Lambda/\Lambda') a_n(g, D)
\end{equation}
and a similar treatment on the effective metric using  minimal subtraction can be found in \cite{C10}.
On the right hand side, in addition to a renormalized eigenvalue $\lambda(\Lambda/\Lambda')$, a counterterm must be added to cancel the one-loop singularity,   
\begin{equation}
   \lambda(\Lambda) = \lambda(\Lambda/\Lambda')+ 2\ln (\Lambda/\Lambda') a_n(g, D)
\end{equation}
Now the eigenvalue $\lambda$ on the left hand side is physical and does not depend on the energy scale and the choice of renormalization schemes, 
that is, its variation with respect to $\tau = 2\ln (\Lambda/\Lambda') $ vanishes.
\begin{equation}
   0 = \frac{\partial \lambda}{\partial \tau}  =  \frac{\partial }{\partial \tau}  \lambda(\tau)+ a_n(g, D)
\end{equation}
In other words, the variation of the eigenvalue with respect to $\tau$ is 
\begin{equation}
   \frac{\partial }{\partial \tau}  \lambda(\tau) = - a_n(g, D)
\end{equation}
Finally, if we introduce a new parameter $t = - \tau$, then  the variation of eigenvalues under the RG flow is
\begin{equation}
   \frac{\partial }{\partial t}  \lambda(t) =  a_n(g, D)
\end{equation}
\end{proof}

From \cite{H77} (or \eqref{RenEffAct}), the renormalized log-determinant can be defined as
\begin{equation}\label{logdet}
  \ln \det(D, \Lambda) = - \zeta'(0, D) - \ln (\Lambda^2) \zeta(0, D)
\end{equation}
for a positive elliptic second-order differential operator $D$. From the proof of the above theorem, 
we know that 
\begin{equation}
   \ln \det  (gD)|_{\Lambda'}^\Lambda \sim - 2\ln (\Lambda/\Lambda') a_n(g, D) = - 2\ln (\Lambda/\Lambda') \zeta(0,g, D)
\end{equation}
In other words, if we set $g = 1$, then the variation of the log-determinant is
\begin{equation}
   \ln \det  (D, \Lambda) - \ln \det(D, \Lambda') = - 2 \ln (\Lambda / \Lambda') \zeta(0, D)
\end{equation}
which agrees with that derived from \eqref{logdet}.

Therefore, under the RG flow, the log-determinant of the Polyakov action varies according to
\begin{equation}
      \frac{\partial }{\partial t}  \ln \det(g, D, t) = a_n(g, D) = \zeta(0, g, D) 
\end{equation}
where $t = - \ln (\Lambda^2)$. 
In Sect. \ref{Conjodd}, we know that the thermal entropy is proportional to the logarithm of the partition function,
\begin{equation}
   S = \ln Z - \beta \frac{\partial }{\partial \beta} \ln Z \, \propto \, \ln Z
\end{equation}
If we  view the nonlinear sigma model \eqref{SigmaMod} as a quantum statistical mechanical system, 
then the  partition function is  the trace of the heat kernel \eqref{Trheat}.
\begin{cor} \label{VarEntrp}
   The variation of the thermal entropy of the Polyakov action is proportional to the top heat kernel coefficient along the RG flow, 
   \begin{equation}
      \frac{\partial S}{\partial t} \propto a_n(g, D) = \zeta(0, g, D)
   \end{equation}

\end{cor}

Hence the variation of eigenvalues really reflects the variation of entropy under the RG flow.  By the feature of heat kernel coefficients,
we have to find a different formula for the variation of eigenvalues in odd dimensions.

\section{A conjecture in odd dimensions} \label{Conjodd}
In this section, we will first go over the thermal entropy and its relation to the partition function and the free energy in statistical mechanics.
After reviewing the irreversibility of the RG flow, i.e., the $c$-, $a$-, and $F$-theorems in different dimensions, we will conjecture a formula
for the variation of eigenvalues of the bulk Laplacian in odd dimensions based on  the geometric  AdS/CFT correspondence \cite{A05, PK05}.

\subsection{Thermal entropy}

Let us first recall some facts in classical and quantum statistical mechanics in the canonical ensemble, where the absolute
 temperature $T$ is  the principal thermodynamic variable and the total energy $E$ may differ between microstates.   
 As usual, the inverse temperature is defined as $\beta = 1/kT$, where $k$ is the Boltzmann constant, and as a convention one always sets $k = 1$.

 In classical statistical mechanics, an ensemble is represented by a probability function $P$ defined over the  phase space.
 According to the  Maxwell--Boltzmann distribution, the probability to find the system at the total energy  $E$ is given by
 \begin{equation}
    P = e^{\beta(F -E)}
 \end{equation}
 where $F$ is the Helmholtz free energy and is a constant for the ensemble.
 Or equivalently, the probability is expressed as
\begin{equation}
      P= \frac{1}{Z} e^{-\beta E} \, \quad \text{with} \quad Z =  e^{-\beta F}
\end{equation}
where $Z$ is called the canonical partition function. 

For example, in a discrete classical canonical ensemble, the canonical partition function is defined as
\begin{equation}
 Z = \sum_{i = 1}^n e^{-\beta E_i}
\end{equation}
where $E_i$ is  the total energy of the system in the respective microstate. 
In this case, the probability $P_i$ of the system occupying the $i$-th microstate  is
\begin{equation}
    P_i= \frac{1}{Z} e^{-\beta E_i}
\end{equation}
and the sum  over the complete set of microstates is $1$, i.e.,
\begin{equation}
   \sum_{i=1}^n P_i = \frac{1}{Z}  \sum_{i=1}^n  e^{-\beta E_i} = 1
\end{equation}
This explains why $P$ is called the probability function and its connection to probability theory, for example see \cite{L12}.

 The free energy $F$ is defined based on the partition function $Z$ by
\begin{equation}
  F= -\frac{1}{\beta}\ln Z
\end{equation} 
In other words,   $F$ can be used interchangeably with $Z$.
The Gibbs entropy is defined by the average of the logarithmic probability,  
\begin{equation}
   S = - \langle \ln P \rangle = \beta \langle E \rangle  + \ln Z
\end{equation}
where $\langle E \rangle $ denotes the average energy.

For convenience, one always expresses the thermodynamic variables in the partition function $Z$ and its derivatives,
\begin{equation}
 F = - \frac{1}{\beta} \ln Z, \quad \langle E \rangle = - \frac{\partial}{\partial \beta}  \ln Z, \quad S = \ln Z - \beta \frac{\partial}{\partial \beta}  \ln Z
\end{equation}
 
The energy fluctuation is the variance of the energy
\begin{equation}
   \sigma = \langle ( E - \langle E \rangle)^2 \rangle  =  \frac{\partial^2}{\partial \beta^2}  \ln Z \geq 0
\end{equation}
and the derivative of the entropy  is
\begin{equation}
    \frac{\partial S }{\partial \beta}=  - \beta \frac{\partial^2}{\partial \beta^2}  \ln Z = -\beta \sigma \leq 0
\end{equation}

To be complete, we also collect the thermodynamic variables expressed in the absolute temperature $T$ and the partition function $Z$,
\begin{equation}
 F = -T \ln Z, \quad \langle E \rangle = T^2 \frac{\partial \ln Z}{ \partial T},  \quad S =  \frac{\partial }{\partial T}(T\ln Z), \quad \frac{\partial S}{\partial T} = \frac{\sigma}{ T^3} \geq 0
\end{equation}
Notice that there is another interesting relation between the free energy $F$ and the thermal entropy $S$,  
\begin{equation}
   S = - \frac{\partial F}{\partial T}
\end{equation}
and $(T, S)$ is a pair of conjugate variables.
In addition, if the average energy vanishes, i.e., $\langle E \rangle = 0 $, then 
\begin{equation}
  S = \ln Z = -\beta F
\end{equation}

In quantum statistical mechanics, a statistical ensemble is represented by a density matrix $\rho$ generalizing the probability function $P$,
\begin{equation}
      \rho = e^{\beta(F -  H)} 
\end{equation}
where $H$ is the Hamiltonian of the canonical ensemble. By the normalization condition $\text{Tr}\, \rho  = 1$, 
one rewrites the above formula as
\begin{equation}
   e^{- \beta F} = \text{Tr} \,e^{- \beta H}
\end{equation}
Or equivalently,  if $Z= \text{Tr} \,e^{- \beta H}$ is  called the partition function as usual, then
\begin{equation}
    F = - \frac{1}{\beta}\ln Z
\end{equation}
Now the von-Neumann entropy is defined as
\begin{equation}
   S = - \text{Tr}\, \rho \ln \rho = - \langle \ln \rho \rangle
\end{equation}

From the relation between quantum statistical mechanics and quantum field theory, the partition function $Z_\beta = Tr (e^{-\beta H})$ will be
effectively identified with the partition function $Z = \int d\phi exp\{ \langle \phi, \Delta \phi\rangle \}$ from the path integral formalism for carefully chosen Hamiltonian $H$ and Laplacian $\Delta$. 
In this case, the inverse temperature $\beta$ will be related to the energy scale by $\beta \sim \Lambda^{-2}$.

\subsection{C-function}
In this subsection, we briefly review the known results about  C-function in low dimensions and  link it to our idea about the thermal entropy and the free energy.
The C-function associated with the RG flow can be viewed as a measure of entropy, which is monotonically decreasing along the evolution.

Ideally, one would like to construct a real-valued positive  function $C(g, \Lambda)$, called the $C$-function,
which  depends on the coupling constant $g$  and the energy scale $\Lambda$ of an effective QFT.
Let $C(g, \Lambda)$ vary along the RG flow; there are two desired properties: 1)  $C(g, \Lambda)$
 decreases monotonically from the UV to IR regime and 2) $C(g, \Lambda)$ is stationary (as a constant independent of $\Lambda$) 
 at the fixed points of the RG flow, i.e., at some conformal field theories (CFTs). 
The first milestone in this direction is Zamolodchikov's $c$-theorem in two dimensions \cite{Z86}. More precisely,
for 2d QFTs, there exists such a $C$-function so that it equals the central charge of the corresponding CFT at some fixed point of the RG flow.
Furthermore, the irreversibility of the RG flow manifests itself in the relation,
\begin{equation}
   c_{UV} > c_{IR}
\end{equation}
Roughly speaking, such formula means high energy theories have more information than theories at low energies.
In essence, the irreversibility of the RG flow is equivalent to the second law of thermaldynamics,
that is, entropy must increase.  

In even dimensions, there is a conformal anomaly (also called trace anomaly or Weyl anomaly, for a review see \cite{D94}), that is,
the trace of the energy-momentum tensor, i.e., $T_\mu^\mu$, is nonzero due to the conformal noninvariance after quantization.
The conformal anomaly can be computed by one-loop approximation and zeta function regularization in terms of heat kernel coefficients,
\begin{equation}
   T^\mu_\mu(x) = a_n(x, D)
\end{equation}
for a Laplace-type operator $D$ \cite{V03}.

In particular,  the conformal anomaly in 2d can be locally expressed as
\begin{equation}
  T^\mu_\mu(x) = a_2(x, D) = \frac{1}{24\pi}R 
\end{equation}
While in a 2d CFT, it is usually written as 
\begin{equation}
   \langle T_\mu^\mu \rangle = - \frac{c}{24\pi}R 
\end{equation}
where $c$ is the central charge. 
The irreversibility of the RG flow, i.e., $ c_{UV} > c_{IR} $, was proved based on computations involving the energy-momentum tensor.
Later it was also interpreted as the entanglement entropy of an interval, since the central charge also appears in the von-Neumann entropy.

Next in 4d, Cardy first  suggested to define the C-function as the conformal anomaly \cite{C88}, whose local form is given by
\begin{equation}
   T^\mu_\mu(x) = a_4(x, D) =  \frac{1}{2880\pi^2 } (R^{ijkl}R_{ijkl} - R^{ij}R_{ij}  -30\Delta R+5R^2 )
\end{equation}
In CFT, it is always written as
\begin{equation}
 \langle T_\mu^\mu \rangle = -\frac{1}{16\pi^2} (aE_4 - cW^2)
\end{equation}
where $E_4$ is the Euler density and $W^2$ is the Weyl tensor squared,
\begin{equation}
   E_4 = R_{\mu\nu\rho\sigma}^2 - 4 R_{\mu\nu}^2 + R^2, \quad W^2_{\mu\nu\rho\sigma}= R_{\mu\nu\rho\sigma}^2 - 2 R_{\mu\nu}^2 + \frac{1}{3}R^2 
\end{equation}
Again, $a$ and $c$ are two central charges in CFT. For conformally flat spaces such as the four-sphere, one has $W^2 = 0$, so  
in this case the conformal anomaly is only depending on the central charge $a$,
\begin{equation}
   \langle T_\mu^\mu \rangle = -\frac{a}{16\pi^2} E_4 
\end{equation}
As an analogy of the $c$-theorem in 2d, there is a so-called $a$-theorem in 4d \cite{KS11}, that is, along the RG flow,
\begin{equation}
   a_{UV} > a_{IR}
\end{equation}

In odd dimensions, there is no conformal anomaly, and one cannot construct a geometric invariant from the Riemann curvature tensor and its derivatives. 
However, under the holographic RG flow, an $F$-theorem was proved in 3d \cite{CHM11}, that is, along the RG flow,
\begin{equation}
   F_{UV} > F_{IR}
\end{equation} 
where $F$ is the free energy of the three sphere,
\begin{equation}
    F = -\frac{1}{\beta} \ln Z_{S^3} = -T  \ln Z_{S^3}
\end{equation}
Later on $F$ was also interpreted as the (renormalized) entanglement entropy in \cite{LM13}. 

The formulation of the 2d c-theorem and 4d a-theorem agrees with our observation that the top heat coefficient (or the conformal anomaly) reflects the variation of entropy (stated in Corollary \ref{VarEntrp}).
The existence of the 3d F-theorem teaches us the free energy plays a similar role as entropy in odd dimensions, which is the key assumption behind our treatment of the odd-dimensional case in the next subsection.

\subsection{Holographic principle}

In the study of renormalization group (RG) flow, the holographic principle is applied to  the holographic RG flow \cite{S02} in the AdS/CFT correspondence \cite{M98, GKP98, W98}.
In addition, entanglement entropy is used to probe the RG flow in the theory space of quantum field theories (QFTs) \cite{CHM11}, note that 
entanglement entropy is the same as thermal entropy for CFTs.
In this section, we will conjecture a formula for the variation of eigenvalues of the bulk Laplacian in odd dimensions 
through the holographic renormalization  in the geometric AdS/CFT correspondence \cite{A05, PK05}.

The AdS/CFT correspondence,  e.g., $AdS_5/CFT_4$,  states that a  gravity theory on an $(n + 1)$-dimensional AdS bulk spacetime is equivalent to a conformal field theory
at the $n$-dimensional boundary of the AdS spacetime. The holographic RG \cite{DVV00, S02} is based on the idea that the radial coordinate of the bulk
with  AdS geometry can be identified with the renormalization parameter in the RG flow  of the boundary field theory.

The geometric setup of the AdS/CFT correspondence is the following,  see \cite{A05} for more details. Let $\bar{X}$ be a   compact   $(n+1)$-manifold, $X$ its interior and $\partial X$ its boundary.
$(X, G)$ models the bulk geometry and the complete Riemannian metric $G$ is assumed to be conformally compact, that is, there is a defining function $\rho$ on $\bar{X}$ such that the conformally equivalent metric
$\tilde{G} = \rho^2G$ extends to a metric on the compactification $\bar{X}$. A defining function $\rho$ is a smooth, non-negative function on $\bar{X}$ with $\rho^{-1}(0) = \partial X$ and $d\rho \neq 0$ on $\partial X$.
The induced metric $g = \tilde{G}|_{\partial X}$ is the boundary metric associated to $\tilde{G}$, and $(\partial X, g) = (M, g)$ models the boundary geometry, which is assumed to be an $n$-dimensional closed manifold. 
There are many possible defining functions, and the conformal compactification $\tilde{G}$ is not unique,  but the conformal class $[g]$ of $g$ is uniquely determined by $(X, G)$, and $(M, [g])$ is called the conformal boundary. 
In  the AdS/CFT correspondence, the bulk geometry $(X, G)$ describes a gravity theory, so we assume $X$ is a (negative) Einstein manifold and $G$ is always taken as a conformally compact Einstein metric. 
From now on, we also assume the bulk geometry $(X, G)$ is $(2k+1)$-dimensional, and the conformal boundary $(M, [g])$ is $2k$-dimensional, i.e., $n= 2k$. 

The Gauss Lemma gives the splitting of the compactification $\tilde{G}$, 
\begin{equation}
   \tilde {G} = d\rho^2 + g_\rho
\end{equation}
where $\rho$  is the unique geodesic defining function determined by the boundary metric $g$,  $g_\rho$ has a Taylor-type series expansion, 
\begin{equation}
  g_\rho \sim g + \rho^2 g_{(2)} + \cdots + \rho^{2k}g_{(2k) } +   \rho^{2k} \ln \rho \, h + \rho^{2k+1}  g_{(2k+1) } + \cdots 
\end{equation}
here $h$ is determined by the conformal anomaly of the boundary CFT.
When the radial coordinate $\rho$ (also called the holographic parameter) tends to $0$, i.e., at the boundary $M$, $g_\rho$ is  equivalent to the boundary metric $g$, 
\begin{equation}
   g_\rho (\rho \rightarrow 0) \sim  g 
\end{equation}
If one truncates the  bulk at an infinitesimal scale  $\rho = \varepsilon$ , then by the holographic renormalization \cite{DVV00, S02}, 
$\rho$ corresponds to the  energy scale $\Lambda \gg 0$ in the RG flow of the boundary QFTs, i.e., the IR/UV connection $\rho \leftrightarrow \Lambda^{-1}$.

\begin{prop} In the geometric setup of the AdS/CFT correspondence, through the holographic RG flow  the free energy of the bulk changes  as 
   \begin{equation}
      \frac{\partial F}{ \partial \rho } =  \frac{1}{2}\zeta'(0, D) -\rho \zeta(0, D) 
   \end{equation} 
   if the zeta function regularization is used, where $D$ is a  Laplace-type operator on the boundary $(M,g)$. 
\end{prop}

\begin{proof}
  Under the AdS/CFT correspondence, we consider the bulk-boundary duality. The basic assumption of the holographic principle
  is that the partition functions from the bulk and boundary are the same 
  \begin{equation}
         Z_{CFT} = Z_{bulk}
  \end{equation}
  where   $Z_{CFT} =  Z_{CFT} ([g])$ only depends on the conformal class of the  boundary metric  $[g]$. 
  
  Hence the free energy of the odd-dimensional bulk theory is 
  \begin{equation}
       F = -T \ln Z_{bulk} = -T \ln Z_{CFT} = -\tau \ln Z_{CFT}
  \end{equation}
  where the absolute temperature $T$ is identified with the parameter $\tau = \ln \Lambda^2$ in the RG flow
  since $T$ has units  in energy. 
  Using the zeta function regularization, the boundary partition function is given by the renormalized one-loop effective action,
  \begin{equation}
        \ln Z_{CFT} = -\frac{1}{2} \ln \det D = \frac{1}{2}[\zeta'(0, D) + \tau \zeta(0, D)]
  \end{equation}
  So the free energy of the bulk is written as
  \begin{equation}
       F =  -\frac{1}{2} [\tau \zeta'(0, D) + \tau^2 \zeta(0, D)]
  \end{equation}
  Hence the derivative with respect to $ \tau$ is
  \begin{equation}
       \frac{\partial F}{ \partial \tau} = - \frac{1}{2}\zeta'(0, D) - \tau \zeta(0, D)  
  \end{equation}
  If we introduce a new parameter $\rho = -\tau = -2 \ln \Lambda $, then the   variation of $F$ along the RG flow is given by
  \begin{equation}
       \frac{\partial F}{ \partial \rho} =  \frac{1}{2}\zeta'(0, D) - \rho \zeta(0, D)  
  \end{equation}
 where $\rho$ is viewed as the holographic parameter.
\end{proof}

   If the thermal entropy $S$ is proportional to the minus free energy $-F$ in odd dimensions, 
   \begin{equation}
                  S \propto - F
   \end{equation}
   then the variation of entropy of the bulk through the holographic RG flow is 
   \begin{equation}
               \frac{\partial S}{ \partial \rho} \propto - \frac{1}{2}\zeta'(0, D)  + \rho \zeta(0, D)
   \end{equation}
   
   Based on the bulk geometry $(X, G)$ with boundary $\partial X = M$, the Einstein--Hilbert action with cosmological constant is considered as a classical gravity theory, 
   the Hamiltonian approach and Hamilton--Jacobi equations have been investigated in the literature, for example see \cite{DVV00, W98}.  
   One can also consider the nonlinear sigma model of bosonic strings with B-fields on $(X, G)$, with a carefully chosen boundary condition $\mathcal{B}$
   the Lagrangian can be converted to a Laplace-type operator $\tilde{D}(G, \mathcal{B})$ \cite{V03}. For simplicity, here we only consider a scalar Laplacian $\tilde{D}(G, \mathcal{B})$   
   (of a scalar field theory) under the Dirichlet or modified Neumann boundary condition $\mathcal{B}$ \cite{V03}, a simple example is given by the free massive scalar field in the AdS geometry with
   \begin{equation}
    \tilde{D} = -\Delta_{G} + m^2
   \end{equation}
   where $\Delta_G$ is the Laplace--Beltrami operator with respect to $G$ \cite{S02}.
   
   If we consider eigenvalues of the Laplace--Beltrami operator $\tilde{D} = -\Delta_{G}$ representing the thermal entropy or the minus free energy of the bulk in  the geometric AdS/CFT correspondence,
   through the holographic renormalization,  we have the following conjecture about the variation of eigenvalues of  $\tilde{D}$  in odd-dimensional bulk.  
\begin{conj}
   In the geometric setup of the AdS/CFT correspondence, if the bulk Laplacian  $\tilde{D} = - \Delta_G$  and the boundary Laplacian $D = -\Delta_g$ are the Laplace--Beltrami operators,  
   then through the holographic renormalization, 
   the variation of eigenvalues of $\tilde{D}$  is
   \begin{equation}
               \frac{\partial \lambda}{ \partial \rho} = - \frac{1}{2}\zeta'(0, D)  +\rho \zeta(0, D)
   \end{equation}
   In particular, if we only approximate the above formula to the first term, then the variation of eigenvalues is
   \begin{equation}
      \frac{\partial \lambda}{ \partial \rho} = - \frac{1}{2}\zeta'(0, D)
   \end{equation}
\end{conj}

\section{Discussion} \label{Disc}
   The irreversibility of the renormalization group (RG) flow is closely related to the monotonicity of entropy.
   And the fixed points of the RG flow, i.e., the corresponding conformal field theories (CFTs), play an important role
   in understanding such irreversibility.
   The eigenvalues derived from entropy functionals inherit the monotonicity of entropy. 
   From the one-loop approximation of a quantum field theory (QFT), 
   the effective partition function can be computed based on  heat kernel expansion and zeta function regularization.
   However, a general heat kernel coefficient does not have the monotonicity, neither do the eigenvalues. 
   In order to establish the monotonicity, one has to decompose a heat
   kernel coefficient into different components. For example, it is possible to find the monotonicity from the Euler density in 4d. 
   
   In odd dimensions, 
   there is no conformal anomaly, so it is only possible to find the monotonicity based on the holographic principle.
   We have learned something interesting from the $F$-theorem in 3d. 
   The free energy in odd dimensions plays the same role as the partition function in even dimensions. 
   There are too many mysterious things that we do not understand in the holographic principle and the AdS/CFT correspondence, 
   so we end this paper with a conjecture.
   We conjectured a formula for the variation of eigenvalues of the odd-dimensional bulk Laplacian in terms of the zeta function of the boundary Laplacian.

   In this paper, we are mainly interested in the Polyakov action \eqref{SigmaMod} and the relevant Laplacian \eqref{LapOp}. 
   It is straightforward to generalize the Polyakov action to include a dilaton field, and it may take more effects to include $B$-fields and other fields in string theory. 
   It is also interesting to apply a similar strategy to study spinor fields or gauge fields evolving along the RG flow.  
   
   We also hope to generalize this work into the noncommutative world. 
   In noncommutative geometry, a Dirac operator in a spectral triple plays the role of a metric in Riemannian geometry. 
   Instead of deformations of a Dirac operator, it is easier to consider  deformations of eigenvalues of the Dirac Laplacian 
   under some geometric flow or the RG flow. 
   There are a lot of interesting things one could try along this direction in noncommutative geometry.


\nocite{*}
\bibliographystyle{plain}
\bibliography{RGF}

\end{document}